\documentclass[11pt]{article}

\usepackage[T1]{fontenc}
\usepackage[margin=1in]{geometry}
\usepackage{amsmath,amssymb,amsthm,mathtools}
\usepackage{aliascnt}
\usepackage{newpxtext,newpxmath}
\usepackage{microtype}
\usepackage{enumitem}
\usepackage{xcolor}
\usepackage[numbers]{natbib}
\usepackage{tikz}
\usetikzlibrary{arrows.meta,patterns}
\usepackage{booktabs}
\usepackage[colorlinks=true,
            linkcolor=blue!55!black,
            citecolor=blue!55!black,
            urlcolor=blue!55!black]{hyperref}
\usepackage[nameinlink,capitalise,noabbrev]{cleveref}

\allowdisplaybreaks
\setlist[itemize]{leftmargin=2em,itemsep=2pt,topsep=4pt}
\setlist[enumerate]{leftmargin=2.4em,itemsep=2pt,topsep=4pt}

\newtheorem{theorem}{Theorem}[section]
\newaliascnt{lemma}{theorem}
\newtheorem{lemma}[lemma]{Lemma}
\aliascntresetthe{lemma}
\newaliascnt{proposition}{theorem}
\newtheorem{proposition}[proposition]{Proposition}
\aliascntresetthe{proposition}
\newaliascnt{corollary}{theorem}
\newtheorem{corollary}[corollary]{Corollary}
\aliascntresetthe{corollary}
\newaliascnt{claim}{theorem}

\aliascntresetthe{claim}
\theoremstyle{definition}
\newaliascnt{definition}{theorem}

\aliascntresetthe{definition}
\newaliascnt{example}{theorem}

\aliascntresetthe{example}
\theoremstyle{remark}
\newaliascnt{remark}{theorem}
\newtheorem{remark}[remark]{Remark}
\aliascntresetthe{remark}

\AddToHook{env/theorem/begin}{\crefalias{section}{theorem}}
\AddToHook{env/lemma/begin}{\crefalias{section}{lemma}}
\AddToHook{env/proposition/begin}{\crefalias{section}{proposition}}
\AddToHook{env/corollary/begin}{\crefalias{section}{corollary}}
\AddToHook{env/claim/begin}{\crefalias{section}{claim}}
\AddToHook{env/definition/begin}{\crefalias{section}{definition}}
\AddToHook{env/example/begin}{\crefalias{section}{example}}
\AddToHook{env/remark/begin}{\crefalias{section}{remark}}

\crefname{theorem}{Theorem}{Theorems}
\Crefname{theorem}{Theorem}{Theorems}
\crefname{lemma}{Lemma}{Lemmas}
\Crefname{lemma}{Lemma}{Lemmas}
\crefname{proposition}{Proposition}{Propositions}
\Crefname{proposition}{Proposition}{Propositions}
\crefname{corollary}{Corollary}{Corollaries}
\Crefname{corollary}{Corollary}{Corollaries}
\crefname{claim}{Claim}{Claims}
\Crefname{claim}{Claim}{Claims}
\crefname{definition}{Definition}{Definitions}
\Crefname{definition}{Definition}{Definitions}
\crefname{example}{Example}{Examples}
\Crefname{example}{Example}{Examples}
\crefname{remark}{Remark}{Remarks}
\Crefname{remark}{Remark}{Remarks}

\newcommand{\E}{\mathbb{E}}
\newcommand{\Prb}{\mathbb{P}}

\newcommand{\1}{\mathbf 1}
\newcommand{\cF}{\mathcal{F}}
\newcommand{\cD}{\mathcal{D}}
\newcommand{\FB}{\operatorname{FB}}

\newcommand{\GFT}{\operatorname{GFT}}
\newcommand{\GSOM}{\operatorname{GSOM}}
\newcommand{\GBOM}{\operatorname{GBOM}}

\newcommand{\OPTS}{\operatorname{OPT}_{S}}
\newcommand{\OPTB}{\operatorname{OPT}_{B}}
\newcommand{\dd}{\,\mathrm d}
\newcommand{\pos}[1]{\left(#1\right)_{+}}
\newcommand{\vct}[1]{\boldsymbol{#1}}

\title{The Power of Simple Mechanisms: A Tight \(e\)-Approximation for Gains from Trade in Matching Markets}
\author{
Xiaohui Bei\thanks{School of Physical and Mathematical Sciences, Nanyang Technological University, Singapore.\\
Email: \href{mailto:xhbei@ntu.edu.sg}{\nolinkurl{xhbei@ntu.edu.sg}}.}
\and
Wenhao Wu\thanks{School of Physical and Mathematical Sciences, Nanyang Technological University, Singapore.\\
Email: \href{mailto:wenhao008@e.ntu.edu.sg}{\nolinkurl{wenhao008@e.ntu.edu.sg}}.}
\and
Shengwei Zhou\thanks{School of Physical and Mathematical Sciences, Nanyang Technological University, Singapore.\\
Email: \href{mailto:s.arthur.zhou@gmail.com}{\nolinkurl{s.arthur.zhou@gmail.com}}.}
}
\date{}

\begin{document}
\maketitle

\begin{abstract}
We study how well simple mechanisms approximate gains from trade (GFT) in two-sided matching markets with independent buyer values and seller costs, where feasible outcomes form an arbitrary downward-closed family of matchings.
This model includes bilateral trade and double auctions as special cases.
We focus on the Generalized Random-Offerer (GRO) mechanism, which is an equal mixture of the Generalized Sellers-Offering Mechanism (GSOM) and the Generalized Buyers-Offering Mechanism (GBOM).
We determine GRO's exact worst-case approximation ratio with respect to first-best GFT, showing that it is $e$.
This improves the previous $3.15$ approximation guarantee~\cite{BabaioffRubinsteinTanWang2026} for the same mechanism.
In bilateral trade, the result implies a $1/e$ guarantee for the random-offerer mechanism, improving the previous $1/\pi$ bound~\cite{Jo2026}.
Our analysis uses a two-dimensional inequality in quantile space to compare first-best GFT with optimal one-sided auction profits.
For each potential trade, we identify the region of buyer values and seller costs for which the first-best allocation selects that trade, and then randomly shrink this region in quantile space, yielding posted-price rules whose expected profits can be evaluated exactly.
We establish tightness of the $e$ approximation by constructing markets with regular type distributions, pairwise disjoint trading edges, and a single knapsack constraint.
On these instances, GRO's expected GFT approaches a $1/e$ fraction of first-best GFT, therefore ruling out any better guarantee even under these restrictions.

\end{abstract}

\section{Introduction}
Two-sided markets play a central role in modern economics, providing a framework for understanding how trade between buyers and sellers allocates resources and creates economic value.
Their scope extends from the simplest exchange of a single item between one buyer and one seller, known as \emph{bilateral trade}, to more general settings involving multiple items and many participants on both sides of the market.
In these larger markets, trading opportunities compete, and the choice of who trades with whom determines how much of the available surplus can be realized.
The goal is to maximize the total surplus generated by a feasible set of trades, measured as the sum of buyers' values minus sellers' costs and called the \emph{gains from trade} (GFT).

When all values and costs are known, the expected maximum GFT achievable is known as the \emph{first-best} benchmark, which captures the market's full potential for surplus generation in the absence of informational constraints.
Achieving this benchmark becomes difficult when values and costs are private.
A mechanism must provide incentives for truthful reporting, ensure that agents are not worse off by participating, and avoid requiring an outside subsidy.
These requirements are captured by Bayesian incentive compatibility (BIC), individual rationality (IR), and weak budget balance (WBB).
Under standard conditions, the classic Myerson--Satterthwaite theorem shows that first-best efficiency is incompatible with BIC, interim IR, and ex-ante WBB even in the simplest setting of \emph{bilateral trade}, where there is only one buyer and one seller trading a single item~\cite{MyersonSatterthwaite1983}.
This impossibility motivates the study of simple mechanisms that satisfy these requirements while guaranteeing a fraction of first-best GFT.

A long line of work has studied first-best approximation in bilateral trade.
Deng et al.~\cite{DengMaoSivanWang2022} established the first constant-factor approximation to first-best GFT for arbitrary independent distributions, proving a $1/8.23$ guarantee for the \emph{random-offerer mechanism}.
This mechanism selects the buyer or the seller with equal probability and lets the selected agent make a utility-maximizing take-it-or-leave-it offer to the other agent.
It is BIC, ex-post IR, and ex-post strongly budget balanced~\cite{BrustleCaiWuZhao2017}.
Subsequent progress came from sharper analyses of this same mechanism: Fei~\cite{Fei2022} improved the guarantee to $1/3.15$, Hartline and Wang~\cite{HartlineWang2025} recovered this bound through a geometric analysis, and Jo~\cite{Jo2026} improved it further to $1/\pi$.

Comparing to bilateral trade, \emph{matching markets} is a more general setting, where multiple buyers and sellers interact and only certain trades are feasible. 
Allowable trades are represented by the edges of a bipartite graph, and a feasible outcome is a matching.
Extending these bilateral guarantees to matching markets requires selecting compatible trades among multiple buyers and sellers and is significantly more challenging.
Br\"ustle et al.~\cite{BrustleCaiWuZhao2017} introduced the Generalized Sellers-Offering Mechanism (GSOM) and the Generalized Buyers-Offering Mechanism (GBOM).
We refer to their equal mixture as the \emph{Generalized Random-Offerer (GRO)} mechanism
\footnote{Babaioff et al.~\cite{BabaioffCaiGonczarowskiZhao2018} refer to the same mechanism as the Random Virtual-Welfare Maximizing (RVWM) mechanism.}.
GRO is dominant-strategy incentive compatible (DSIC), ex-post IR, and ex-ante WBB.
In bilateral trade, its allocation rule coincides with that of the random-offerer mechanism.
Obtaining a constant-factor approximation to first-best GFT in matching markets had remained a major open question in two-sided markets~\cite{CaiGoldnerMaZhao2021}.
Only very recently, Babaioff et al.~\cite{BabaioffRubinsteinTanWang2026} established the first constant-factor approximation to first-best GFT in general matching markets by proving a $1/3.15 \approx 0.317$ guarantee for this mechanism, even under arbitrary downward-closed feasibility constraints.

These results establish constant-factor lower bound guarantees.
In the other direction, upper bounds, which indicate the limitation of the mechanism's performance, also arise already in bilateral trade.
Cai et al.~\cite{CaiGoldnerMaZhao2021} and, independently, Babaioff, Dobzinski, and Kupfer~\cite{BabaioffDobzinskiKupfer2021} ruled out a $1/2$ guarantee for the random-offerer mechanism, with the latter construction giving a first-best fraction below $0.495$.
Subsequent counterexamples lowered the upper bound on its worst-case guarantee to $0.481951$~\cite{CaiGuptaLiMehta2026} and then $0.460243$~\cite{Jo2026}.
Since bilateral trade is a special case of matching markets, these examples also give upper bounds on the worst-case guarantee of GRO in general matching markets.
However, none of these upper bounds is tight, and a gap remains between the best known lower and upper bounds for GRO in matching markets.

We close this gap by completely characterizing GRO's worst-case first-best approximation guarantee over downward-closed matching markets: the exact fraction is $1/e$.
That is, GRO always obtains at least a $1/e$ fraction of first-best GFT, and this fraction cannot be improved over the full class of downward-closed matching markets.
As a consequence, we also improve the first-best guarantee of the random-offerer mechanism in bilateral trade from $1/\pi$ to $1/e$.

\subsection{Our Results}\label{subsec:results}

We consider matching markets with independent buyer values and seller costs and an arbitrary downward-closed family of feasible matchings.
Let $\mathcal M$ denote GRO, i.e., the mechanism that runs GSOM or GBOM with equal probability.
It is DSIC, ex-post IR, and ex-ante WBB~\cite{BrustleCaiWuZhao2017}.
Our main result determines its exact worst-case approximation guarantee.

\begin{theorem}[Main theorem, informal]\label{thm:main-informal}
In matching markets with independent types and downward-closed feasibility, $\mathcal M$ obtains at least a $1/e$ fraction of first-best GFT.
Moreover, this guarantee is tight for $\mathcal M$ over this class of markets.
\end{theorem}

This improves the $1/3.15$ guarantee of Babaioff et al.~\cite{BabaioffRubinsteinTanWang2026} and determines the exact worst-case guarantee of GRO.
We establish tightness by constructing instances with regular type distributions supported on $[0,1]$.
All trading edges have pairwise disjoint endpoints, and feasibility is imposed by a single knapsack constraint.
Along this family, the ratio of GRO's expected GFT to first-best GFT approaches $1/e$.
The formal statements appear in \cref{thm:main,thm:tightness}.

Our guarantee also improves the approximation achievable in bilateral trade.
With one buyer and one seller, GSOM and GBOM have the same allocation rules as the seller-offering and buyer-offering mechanisms, respectively.
Thus GRO reduces to the random-offerer mechanism in bilateral trade.

\begin{corollary}[Random-offerer mechanism, informal]\label{cor:random-offerer}
In bilateral trade with independent buyer values and seller costs, the random-offerer mechanism obtains at least a $1/e$ fraction of first-best GFT.
\end{corollary}

This improves the previous $1/\pi$ guarantee of Jo~\cite{Jo2026}.
Together with its upper bound, our result places the worst-case first-best fraction of the random-offerer mechanism in $[1/e,0.460243)$.

\subsection{Technical overview}\label{subsec:overview}

Our approximation proof follows the one-sided profit approach of Babaioff et al.~\cite{BabaioffRubinsteinTanWang2026}.
For GRO's GSOM component, the seller-side profit benchmark counts buyer payments net of the matched sellers' costs; for its GBOM component, the buyer-side benchmark counts the value of the matched buyers net of seller payments.
Let $\Pi_S(\GSOM)$ and $\Pi_B(\GBOM)$ denote these expected profits, and let $\FB$ denote expected first-best GFT.
We prove
\begin{equation}\label{eq:main-profit}
  \FB
  \le \frac e2\Bigl(\Pi_S(\GSOM)+\Pi_B(\GBOM)\Bigr).
\end{equation}
Since the expected GFT of each component is at least its corresponding one-sided profit, this inequality implies the $1/e$ guarantee for GRO.
The main task is to compare first-best GFT with local posted-price profits and then realize these profits through feasible one-sided auctions.

\paragraph{From matching markets to local trade regions.}
Fix a trading edge and all reports outside its endpoints.
The remaining buyer value and seller cost are independent, but conditioning on the edge being selected by first best generally correlates them.
Consequently, a guarantee for independent bilateral trade does not apply directly to the selected trade.

Following Bei et al.~\cite{BeiLiWuZhou2026SecondBest}, we represent this selection event as a region in quantile space.
Write the two types as $Q_B(\rho)$ and $Q_S(\sigma)$, where $\rho$ and $\sigma$ are independent uniform buyer upper-tail and seller lower-tail ranks.
Smaller ranks correspond to higher buyer values and lower seller costs.
The region $\cD$ on which first best selects the edge is a down-set: whenever it contains a rank pair, it also contains every coordinatewise smaller pair.
Averaging the GFT contribution of this region over the fixed reports and summing over edges recovers first-best GFT.

To bound this contribution, we compare the GFT on $\cD$ with the profits from posted prices whose accepted trades lie in $\cD$.
Let $\ell_B(\sigma)$ and $\ell_S(\rho)$ denote its horizontal and vertical section lengths.
For a seller with rank $\sigma$, a price serving the top $x<\ell_B(\sigma)$ buyer mass yields expected profit $x(Q_B(x)-Q_S(\sigma))$.
Optimizing separately on each horizontal section and averaging gives
\[
  P_S(\cD):=\int_0^1
    \sup_{0\le x<\ell_B(\sigma)}
    x\bigl(Q_B(x)-Q_S(\sigma)\bigr)\dd \sigma,
\]
where $x=0$ denotes no trade and empty sections contribute zero.
Define $P_B(\cD)$ symmetrically using the vertical sections.
Our main analytic bound is
\begin{equation}\label{eq:intro-local}
  \int_\cD\bigl(Q_B(\rho)-Q_S(\sigma)\bigr)\dd \rho\dd \sigma
  \le \frac e2\bigl(P_S(\cD)+P_B(\cD)\bigr).
\end{equation}
The comparison must respect the section-dependent restrictions on admissible prices.

\paragraph{Proving the local inequality.}
Draw $h$ uniformly from $(0,1)$ and consider the complementary contraction
\begin{equation*}
  (\rho,\sigma)\mapsto\bigl(e^{-h}\rho,e^{-(1-h)}\sigma\bigr),
  \qquad h\in(0,1).
\end{equation*}
The horizontal and vertical boundaries of the contracted region specify two posted-price rules.
Because $\cD$ is a down-set, their thresholds lie inside the original sections.
Choose the horizontal rule with probability $1-h$ and the vertical rule with probability $h$.

We evaluate the resulting expected one-sided profit by separating the buyer and seller contributions.
The integral identity in \cref{lem:exponential-scale} shows that the averaged buyer contribution---buyer payments under the horizontal rule and buyer values under the vertical rule---equals $1/e$ of the original buyer-value integral.
The symmetric calculation gives $1/e$ of the original seller-cost integral.
Subtracting these identities shows that the expected local profit is exactly $1/e$ of the original GFT.
Each type of rule is selected with total probability $1/2$, and its profit is at most the corresponding section benchmark, yielding \eqref{eq:intro-local}.
The argument applies directly to integrable quantiles, including distributions with point masses.
Figure~\ref{fig:local-geometry} illustrates the contraction.

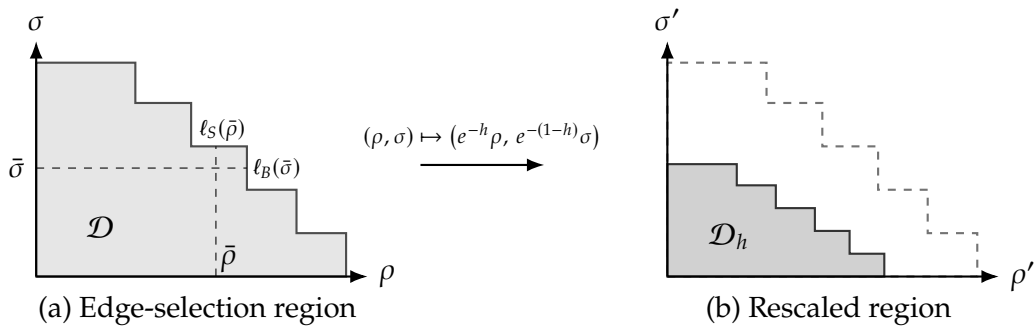
\begin{figure}[h]
\centering
\begin{tikzpicture}[x=0.82cm,y=0.82cm,>=Latex,font=\normalsize]
  \begin{scope}
    \path[fill=black!10,draw=black!70,thick]
      (0,0)--(5.0,0)--(5.0,0.7)--(4.2,0.7)--(4.2,1.4)--(3.4,1.4)--(3.4,2.1)--(2.5,2.1)--(2.5,2.8)--(1.6,2.8)--(1.6,3.45)--(0,3.45)--cycle;
    \draw[->,thick] (0,0)--(5.35,0) node[right] {$\rho$};
    \draw[->,thick] (0,0)--(0,3.8) node[above] {$\sigma$};

    \draw[dashed] (0,1.75)--(3.4,1.75);
    \draw[dashed] (2.9,0)--(2.9,2.1);
    \node[left] at (0,1.75) {$\bar \sigma$};
    \node[above] at (3.1,-0.1) {$\bar \rho$};
    \node[font=\scriptsize,right,inner sep=2pt] at (3.4,1.75) {$\ell_B(\bar \sigma)$};
    \node[font=\scriptsize,above,inner sep=2pt] at (3.0,2.1) {$\ell_S(\bar \rho)$};

    \node at (1.05,0.8) {$\cD$};
    \node at (2.6,-0.55) {(a) Edge-selection region};
  \end{scope}

  \draw[->,thick] (6.2,1.8)--(8.2,1.8);
  \node[font=\scriptsize,above,inner sep=2pt] at (7.2,1.95)
    {$ (\rho,\sigma)\mapsto\bigl(e^{-h}\rho,\,e^{-(1-h)}\sigma\bigr) $};

  \begin{scope}[xshift=8.35cm]
    \path[draw=black!55,dashed,thick]
      (0,0)--(5.0,0)--(5.0,0.7)--(4.2,0.7)--(4.2,1.4)--(3.4,1.4)--(3.4,2.1)--(2.5,2.1)--(2.5,2.8)--(1.6,2.8)--(1.6,3.45)--(0,3.45)--cycle;
    \begin{scope}[xscale=0.70,yscale={exp(-1)/0.70}]
      \path[fill=black!18,draw=black!80,thick]
        (0,0)--(5.0,0)--(5.0,0.7)--(4.2,0.7)--(4.2,1.4)--(3.4,1.4)--(3.4,2.1)--(2.5,2.1)--(2.5,2.8)--(1.6,2.8)--(1.6,3.45)--(0,3.45)--cycle;
    \end{scope}
    \draw[->,thick] (0,0)--(5.35,0) node[right] {$\rho'$};
    \draw[->,thick] (0,0)--(0,3.8) node[above] {$\sigma'$};
    \node at (1.0,0.65) {$\cD_h$};
    \node at (2.6,-0.55) {(b) Rescaled region};
  \end{scope}
\end{tikzpicture}
\caption{The edge-selection region before and after the rescaling.
The product of the two scale factors is $e^{-1}$.}
\label{fig:local-geometry}
\end{figure}

\paragraph{Realizing the local profits.}
The local prices were optimized separately, so implementing them simultaneously requires an additional feasibility argument.
We use the stable-partner property proved by Babaioff et al.~\cite[Lemma~4.4]{BabaioffRubinsteinTanWang2026}: holding all other reports fixed, an agent has the same partner whenever she is matched by first best.
Thus each buyer faces at most one active offer, with a price independent of her own report.
This gives a truthful one-sided auction for the buyers; the seller-side construction is symmetric.

Every accepted offer belongs to the first-best matching at the realized profile.
The accepted edges therefore form a subset of one feasible matching, and downward closedness preserves feasibility, as in the assembly argument of Bei et al.~\cite{BeiLiWuZhou2026SecondBest}.
Since GSOM and GBOM---the two components of GRO---maximize the corresponding one-sided profits, their profits dominate those of these auxiliary auctions.
Summing the local inequality~\eqref{eq:intro-local} consequently gives the desired market-wide bound~\eqref{eq:main-profit}.

\paragraph{Tightness.}
Our tight examples start from a capped exponential buyer distribution and its reflected seller distribution, for which the local profit comparison approaches equality.
This alone does not establish tightness for GFT, because GFT also includes the buyers' or sellers' utilities.
GSOM and GBOM maximize seller-side and buyer-side \emph{virtual surplus}, obtained from GFT by replacing buyer values with their ironed virtual values or seller costs with their ironed virtual costs, respectively; see \cref{subsec:one-sided}.

We take many independent copies of this bilateral instance and add one deterministic trading edge.
The GFT from this edge lies above the expected maximum virtual surplus of the random edges for either component of GRO, but below their expected maximum GFT.
A single knapsack constraint permits either a subset of the random edges or the deterministic edge alone.
As the number of copies grows, concentration makes first best prefer the profitable random trades, while both components of GRO prefer the deterministic edge, with probability tending to one.
Choosing the GFT from the deterministic edge close to the virtual-surplus benchmark then converts the tight profit comparison into a GFT ratio approaching $1/e$.

\subsection{Other related work}\label{sec:related}

Early first-best approximation results for bilateral trade relied on additional distributional assumptions.
McAfee~\cite{McAfee2008} obtained a $1/2$ guarantee using a fixed price when the median buyer value exceeds the median seller cost.
Under a monotone hazard rate assumption on the buyer's distribution, Blumrosen and Mizrahi~\cite{BlumrosenMizrahi2016} established a $1/e$ guarantee for the seller-offering mechanism, which lets the seller choose a price maximizing her expected utility.
Fei~\cite{Fei2022} subsequently improved this guarantee to the tight bound $1/(e-1)$ under the same assumption.

Deng et al.~\cite{DengMaoSivanWangWu2025} study bilateral trade with limited distributional information, where the agent making the offer has sample access to the other agent's distribution.
Under their sampling assumptions, they prove that the better of the seller-offering and buyer-offering mechanisms gives a constant-factor approximation to first-best GFT when prices approximately maximize profit on the observed samples.
Hajiaghayi et al.~\cite{HajiaghayiHajiaghayiPengShin2025} study bilateral trade with a broker who maximizes expected buyer payments net of seller payments.
They establish constant-factor GFT guarantees under additional distributional assumptions and show that the fraction of first-best GFT can approach zero in general.
Our analysis assumes publicly known distributions and uses profit benchmarks for the buyer or seller side of the market.

\paragraph{Second-best guarantees.}
Another line of work measures simple mechanisms against the \emph{second-best} benchmark, the maximum expected GFT attainable by BIC, interim-IR, and ex-ante WBB mechanisms.
Br\"ustle et al.~\cite{BrustleCaiWuZhao2017} showed that GRO obtains one half of the second-best GFT under downward-closed feasibility.
Babaioff et al.~\cite{BabaioffCaiGonczarowskiZhao2018} combined offering mechanisms with trade reduction, building on McAfee~\cite{McAfee1992}, to obtain both a constant ex-ante approximation to second-best GFT and asymptotic first-best efficiency in double auctions and matching markets.

The efficiency of the second-best relative to the first-best has also been studied.
In bilateral trade, Blumrosen and Mizrahi~\cite{BlumrosenMizrahi2016} established an upper bound of $2/e$ on the worst-case first-best fraction of the second-best.
Liu et al.~\cite{LiuQinRenWang2026} subsequently determined the exact worst-case fraction to be $1/2$.
Bei et al.~\cite{BeiLiWuZhou2026SecondBest} extended this tight $1/2$ guarantee to downward-closed matching markets. Liu et al.~\cite{liu2026bilateraltradematchingmarkets} gave a general bilateral-to-matching reduction for second-best guarantees and used it to establish the worst-case guarantees under MHR and finite-support assumptions.

\paragraph{Multi-dimensional markets.}
Cai et al.~\cite{CaiGoldnerMaZhao2021} obtained an $O(\log^2 n)$ approximation to second-best GFT for a constrained-additive buyer and $n$ single-item sellers with independent item values and costs.
In their setting, the buyer has a multi-dimensional type, while each seller has only a single private cost.
Rubinstein et al.~\cite{RubinsteinTanZhou2026} study markets with multi-dimensional types on both sides: one additive seller owns multiple heterogeneous items, and buyers have XOS valuations, which are maxima of additive functions.
Under independence across item types and costs, they show that the equal mixture of the seller-optimal and buyer-optimal mechanisms obtains a $1/44$ fraction of first-best GFT with one buyer.
They also obtain a constant-factor approximation to first-best GFT with multiple XOS buyers.
Their mechanisms are BIC, interim-IR, and ex-ante WBB.
Their analysis combines a core--tail decomposition with reductions to single-dimensional matching markets, where it uses the same one-sided profit benchmarks that we bound in this paper.

\paragraph{Social welfare.}
A related line of work studies social welfare, which counts the values of both the buyers who receive items and the sellers who keep their items.
In bilateral trade with independent values, Blumrosen and Dobzinski~\cite{BlumrosenDobzinski2021} obtained a $(1-1/e)$-approximation to optimal expected welfare using fixed-price mechanisms.
Subsequent work improved fixed-price welfare guarantees under different assumptions about the available distributional information~\cite{KangPerniceVondrak2022,CaiWu2023,LiuRenWang2023}.
Colini-Baldeschi et al.~\cite{ColiniBaldeschiDeKeijzerLeonardiTurchetta2016,ColiniBaldeschiGoldbergDeKeijzerLeonardiRoughgardenTurchetta2020} obtained constant-factor welfare approximations in double auctions and two-sided combinatorial markets, respectively, using incentive-compatible, individually rational, and strongly budget-balanced mechanisms.
In these settings, welfare equals GFT plus the value of the items initially owned by the sellers, so a constant-factor welfare guarantee need not imply a constant-factor GFT guarantee.

\paragraph{Technical connections.}
D\"utting et al.~\cite{DuttingRoughgardenTalgamCohen2014} relate the incentive, budget-balance, and efficiency properties of double auctions to those of ranking algorithms for the two sides and their composition.
The one-sided auction tools in our proof build on Myerson's virtual-value framework~\cite{Myerson1981}.
The Meta-Auction of Babaioff et al.~\cite{BabaioffRubinsteinTanWang2026} regroups first-best outcomes into bilateral instances with a fixed buyer value and a capped seller distribution, and uses cap-monotonicity to establish buyer incentives.
Here, cap-monotonicity means that, with both agents' reports fixed, any trade made by the bilateral rule still occurs when the seller-cost cutoff increases~\cite[Definition~4.8]{BabaioffRubinsteinTanWang2026}.
Our analysis uses the two-dimensional edge-selection region and bounds its contribution to first-best GFT directly by the profits from prices optimized separately on its horizontal and vertical sections.
We also use the rank-space edge decomposition and down-set contractions of Bei et al.~\cite{BeiLiWuZhou2026SecondBest}, with a different local objective and choice of scales, as detailed in \cref{subsec:overview}.

Jo~\cite[Sections~3.2 and~4]{Jo2026} compares first-best GFT with the sum of the two offerer profits through a Lagrangian analysis, using quantile ironing to remove a seller regularity assumption.
Our local analysis constructs actual posted prices on contracted trade regions and uses the integral identity to evaluate buyer and seller contributions separately.

\section{Preliminaries}\label{sec:model}
We consider a general two-sided matching market with multiple buyers and sellers, trading on a single type of item.
Throughout this paper, we write $z_+:=\max\{z,0\}$ for the positive part of a real number $z$.

\subsection{Model and notation}\label{subsec:model-notation}

Let $B$ be a finite set of buyers, $S$ a finite set of sellers, and $E\subseteq B\times S$ the set of possible trading edges.
An allocation is a matching $M\subseteq E$ such that no agent is incident to more than one edge of $M$.
The feasibility of all allocations is described by a nonempty family $\cF$ of matchings.
Throughout the paper, $\cF$ is downward closed:
\begin{equation*}
  M\in\cF,\quad M'\subseteq M
  \quad\Longrightarrow\quad M'\in\cF.
\end{equation*}
In particular, $\varnothing\in\cF$.
The assumption allows ordinary graph matching together with additional cardinality, knapsack, capacity, or regional constraints, and arbitrary intersections of such restrictions.

Buyer $i$ has private value $v_i\ge0$ for receiving one unit of the item, drawn from a known distribution $F_i$.
Seller $j$ has private cost $c_j\ge0$ for supplying one unit of the item, drawn from a known distribution $G_j$.
All types are mutually independent, their distributions are common knowledge, and every distribution has finite first moment.
Write $\vct t=(\vct v,\vct c)$ for the type profile.
We use $v_i$ and $c_j$ both for random types and for their realizations.
In probabilistic expressions, they are understood to be drawn from $F_i$ and $G_j$, respectively.
For a realized profile $\vct t$ and $M\in\cF$, the gains from trade are
\begin{equation*}
  \GFT(M;\vct t):=\sum_{(i,j)\in M}(v_i-c_j).
\end{equation*}
Fix a strict order on $\cF$ for tie-breaking.
The first-best matching is the unique tie-broken maximizer
\begin{equation*}
  M^*(\vct t)\in\arg\max_{M\in\cF}\GFT(M;\vct t),
\end{equation*}
and its expected GFT is
\begin{equation*}
  \FB:=\E_{\vct t}\bigl[\GFT(M^*(\vct t);\vct t)\bigr].
\end{equation*}
This expectation is finite because the market is finite and every type has finite first moment.

A direct mechanism asks agents to report their types, selects a feasible possibly randomized matching, charges buyers, and pays sellers.
Utilities are quasilinear.
We use dominant-strategy incentive compatibility (DSIC), ex-post individual rationality (IR), and ex-ante weak budget balance (WBB) for the two-sided mechanisms.

DSIC requires that, for every agent and true type, truthful reporting maximizes expected utility for every fixed profile of the other agents' reports, averaging only over the mechanism's internal randomization.
Ex-post IR requires every truthful agent to have nonnegative utility for every profile of the other agents' reports and every realization of the mechanism's randomness.
Ex-ante WBB requires that, under truthful reporting, expected total buyer payments are at least expected total seller payments, averaging over the full type profile and the mechanism's internal randomization.

We use quantile representations for buyer values and seller costs throughout the analysis.
Write $\rho$ for buyer upper-tail ranks and $\sigma$ for seller lower-tail ranks.
For buyer $i$, let $Q_i:(0,1)\to[0,\infty)$ be a nonincreasing left-continuous upper-quantile function: if $\rho_i\sim\operatorname{Unif}(0,1)$, then $Q_i(\rho_i)\sim F_i$.
Thus smaller $\rho$ means a higher buyer value.
For seller $j$, let $Q_j$ be a nondecreasing left-continuous lower-quantile function: if $\sigma_j\sim\operatorname{Unif}(0,1)$, then $Q_j(\sigma_j)\sim G_j$.
Thus smaller $\sigma$ means a lower seller cost.
We couple all types through independent uniform ranks.

For generic buyer and seller quantile functions $Q_B$ and $Q_S$ representing random variables $v$ and $c$, define the revenue and procurement-cost curves
\begin{equation*}
  R(\rho):=\rho Q_B(\rho),\qquad K(\sigma):=\sigma Q_S(\sigma),\qquad R(0)=K(0)=0.
\end{equation*}
At $1$, define both curves by their left limits; $K(1)=+\infty$ is allowed when costs are unbounded.
For $0<x<1$, a threshold serving exactly the top $x$ buyer quantile mass has expected payment $R(x)$.
For $0<y<1$, a threshold procuring from exactly the bottom $y$ seller quantile mass has expected payment $K(y)$.
At atoms, the boundary type is randomized so that exactly the desired mass is served.
Appendix~\ref{app:quantile} records the convention.

\subsection{One-sided benchmarks}\label{subsec:one-sided}

The proof compares first-best GFT with two one-sided optimization problems.
To avoid ambiguity, in the seller-side problem seller costs are treated as parameters and the buyers are strategic; in the buyer-side problem buyer values are treated as parameters and the sellers are strategic.

For the seller-side benchmark, fix a realized seller-cost vector $\vct c$.
A buyer auction selects a feasible matching $M$ and charges buyer payments $p_i$.
Its objective is revenue from the buyers net of the realized seller costs:
\begin{equation*}
  \pi_S(M,\vct p;\vct t)
  :=\sum_{i\in B}p_i-\sum_{(i,j)\in M}c_j.
\end{equation*}
Let $\OPTS$ be the supremum of $\E[\pi_S]$ over jointly measurable families of feasible buyer auctions indexed by the seller-cost vector $\vct c$.
For every fixed $\vct c$, the corresponding auction must be Bayesian incentive compatible (BIC) and interim IR for the buyers.
Here, BIC requires that, for every buyer and true value, truthful reporting maximizes expected utility when the other buyers report truthfully, averaging over their types and the auction's internal randomization.
Interim IR requires this truthful expected utility to be nonnegative for every buyer and true value.
The expectation in $\E[\pi_S]$ is over the full type profile $(\vct v,\vct c)$ and the auction's internal randomization.

Symmetrically, for the buyer-side benchmark, fix a realized buyer-value vector $\vct v$.
A procurement auction selects a feasible matching $M$ and pays the sellers amounts $r_j$.
Its objective is buyer value net of seller payments:
\begin{equation*}
  \pi_B(M,\vct r;\vct t)
  :=\sum_{(i,j)\in M}v_i-\sum_{j\in S}r_j.
\end{equation*}
Let $\OPTB$ be the supremum of $\E[\pi_B]$ over jointly measurable families of feasible procurement auctions indexed by the buyer-value vector $\vct v$.
For every fixed $\vct v$, the corresponding auction must be BIC and interim IR for the sellers, with interim expectations taken over the other sellers' types and the auction's internal randomization.
The expectation in $\E[\pi_B]$ is over the full type profile $(\vct v,\vct c)$ and the auction's internal randomization.

For a two-sided mechanism $\mathcal A$, write $\Pi_S(\mathcal A):=\E[\pi_S(\mathcal A)]$ and $\Pi_B(\mathcal A):=\E[\pi_B(\mathcal A)]$ when evaluating the corresponding one-sided objective.
These are comparison benchmarks; they are not the intermediary's realized budget surplus in the two-sided implementation.

\paragraph{Ironed virtual values and costs.}
For buyer $i$, let $R_i(\rho):=\rho Q_i(\rho)$ and let $\overline R_i$ be the pointwise smallest concave function satisfying $\overline R_i\ge R_i$ on $[0,1]$.
For seller $j$, let $K_j(\sigma):=\sigma Q_j(\sigma)$ and let $\underline K_j$ be the pointwise largest convex function satisfying $\underline K_j\le K_j$ on $[0,1]$, using the extended endpoint convention in \cref{subsec:model-notation} when $K_j(1)=+\infty$.
For $\rho,\sigma\in(0,1)$, let $\Phi_i(\rho)$ and $\Psi_j(\sigma)$ be their interior left derivatives, respectively.
These are the ironed virtual value and virtual cost in quantile space; $\Phi_i$ is nonincreasing and $\Psi_j$ is nondecreasing.

For nonnegative type reports $v,c$, define the type-space ironed virtual functions by
\begin{align*}
  \overline\phi_i(v)&:=\sup\{\Phi_i(\rho):\rho\in(0,1),\ Q_i(\rho)\le v\},\\*
  \overline\psi_j(c)&:=\inf\{\Psi_j(\sigma):\sigma\in(0,1),\ Q_j(\sigma)\ge c\},
\end{align*}
with $\sup\varnothing=-\infty$ and $\inf\varnothing=+\infty$.
These functions are nondecreasing in the reported type and agree almost surely with the corresponding quantile slopes.
They satisfy $\overline\phi_i(v)\le v$ and $\overline\psi_j(c)\ge c$.
They also specify the values for reports in gaps or outside the support, so the threshold payments below are unambiguous.

\paragraph{GSOM and GBOM.}
Using these ironed virtual functions, GSOM and GBOM select, respectively,
\begin{align*}
  M^S(\vct t)&\in\arg\max_{M\in\cF}\sum_{(i,j)\in M}\bigl(\overline\phi_i(v_i)-c_j\bigr),\\
  M^B(\vct t)&\in\arg\max_{M\in\cF}\sum_{(i,j)\in M}\bigl(v_i-\overline\psi_j(c_j)\bigr).
\end{align*}
We refer to these two objectives as seller-side and buyer-side virtual surplus, respectively.
Both rules break ties using the fixed strict order on $\cF$.
An edge with score $-\infty$ is never selected.
Both mechanisms use normalized threshold payments on each side.
Holding all other reports fixed, a matched buyer pays the infimum nonnegative report at which she is matched, and a matched seller receives the supremum nonnegative report at which she is matched.
Unmatched agents make and receive zero payments.
At a threshold report, allocation is determined by the same fixed tie-breaking rule.
A matched seller's threshold payment is at most the highest buyer value, so it is finite even when her cost distribution has unbounded support.

We use the following standard implementation fact.
One-sided optimality is the usual ironed virtual-surplus characterization.
The two-sided DSIC, ex-post-IR, ex-ante-WBB implementations are due to Br\"ustle et al.~\cite{BrustleCaiWuZhao2017}.
Appendix~\ref{app:one-sided} recalls the one-sided argument.

\begin{proposition}[GSOM and GBOM]\label{prop:gsom-gbom}
There are two-sided implementations of GSOM and GBOM satisfying
\begin{equation*}
  \Pi_S(\GSOM)=\OPTS,
  \qquad
  \Pi_B(\GBOM)=\OPTB.
\end{equation*}
Both implementations are DSIC, ex-post IR, and ex-ante WBB.
The statement applies to all type distributions in our model, including distributions with point masses.
\end{proposition}

The next identity is elementary but useful: one-sided profit is obtained from GFT by subtracting the strategic side's utility.

\begin{proposition}\label{lem:gft-profit}
For every seller-side auction that is interim IR for the buyers,
\[
  \E[\GFT]\ge \E[\pi_S].
\]
For every buyer-side auction that is interim IR for the sellers,
\[
  \E[\GFT]\ge \E[\pi_B].
\]
\end{proposition}

\begin{proof}
Let $q_i$ and $q_j$ be the allocation indicators of buyer $i$ and seller $j$, respectively, and let $u_i=q_iv_i-p_i$ be the buyer's realized utility.
Pointwise,
\[
  \GFT(M;\vct t)-\pi_S(M,\vct p;\vct t)=\sum_{i\in B}u_i.
\]
Taking expectations and using interim IR proves the first inequality.
For the buyer-side problem, seller utility is $u_j=r_j-q_jc_j$, and
\[
  \GFT(M;\vct t)-\pi_B(M,\vct r;\vct t)=\sum_{j\in S}u_j.
\]
The second inequality follows in the same way.
\end{proof}

\section{An \texorpdfstring{$e$}{e}-Approximation for Matching Markets}\label{sec:approximation}

In this section, we prove that GRO obtains at least a $1/e$ fraction of first-best GFT in downward-closed matching markets.
Our proof compares first-best GFT with the two one-sided profit benchmarks introduced in \cref{subsec:one-sided}.

\begin{theorem}[Main theorem]\label{thm:main}
Consider a finite Bayesian matching market with publicly known, mutually independent distributions over nonnegative buyer values and seller costs, each with finite first moment.
Suppose the feasible allocations form a downward-closed family of matchings.
Let $\mathcal M$ denote GRO, i.e., the mechanism that runs GSOM or GBOM with equal probability, independently of all reports.
Then
\begin{equation*}
  \E[\GFT(\mathcal M)]\ge \frac1e\,\FB.
\end{equation*}
Moreover, $\mathcal M$ is dominant-strategy incentive compatible, ex-post individually rational, and ex-ante weakly budget balanced.
\end{theorem}

By \cref{prop:gsom-gbom,lem:gft-profit}, it is enough to prove the benchmark inequality
\begin{equation*}
  \FB\le \frac e2(\OPTS+\OPTB).
\end{equation*}
We prove this inequality by reducing first-best GFT to edge-level benchmarks and realizing these benchmarks through feasible one-sided auctions.
We state the key local inequality in \cref{subsec:local-problem} and use it to derive the approximation guarantee in \cref{subsec:implementation}, before proving the inequality in \cref{subsec:local-proof}.

\subsection{First-best edge-selection regions}\label{subsec:edge-regions}

We recall the rank-space representation used by Bei et al.~\cite{BeiLiWuZhou2026SecondBest} and the monotonicity properties of Babaioff et al.~\cite[Lemmas~4.3 and~4.4]{BabaioffRubinsteinTanWang2026}.
We include proofs under our tie-breaking and distributional conventions.

Fix an edge $e=(i,j)$.
Let $\omega_{-e}$ denote the ranks of all agents other than buyer $i$ and seller $j$.
Write $\vct t_{-ij}(\omega_{-e})$ for the corresponding external type profile.
Conditional on $\omega_{-e}$, define the first-best selection region of $e$ by
\begin{equation*}
  \cD_e^{\omega_{-e}}
  :=\Bigl\{(\rho,\sigma):
      e\in M^*\bigl(Q_i(\rho),Q_j(\sigma),\vct t_{-ij}(\omega_{-e})\bigr)
    \Bigr\}.
\end{equation*}

\begin{proposition}[Edge-selection regions]\label{prop:edge-region}
For every edge $e=(i,j)$ and external profile $\omega_{-e}$:
\begin{enumerate}[label=(\roman*)]
  \item $\cD_e^{\omega_{-e}}$ is a measurable down-set;
  \item $Q_i(\rho)\ge Q_j(\sigma)$ for almost every $(\rho,\sigma)\in \cD_e^{\omega_{-e}}$; and
  \item first-best GFT decomposes as
  \begin{equation*}
    \FB
    =\sum_{e=(i,j)\in E}
      \E_{\omega_{-e}}
      \left[
        \int_{\cD_e^{\omega_{-e}}}
          \bigl(Q_i(\rho)-Q_j(\sigma)\bigr)\dd \rho\dd \sigma
      \right].
  \end{equation*}
\end{enumerate}
\end{proposition}

\begin{proof}
Recall that the external ranks $\omega_{-e}$ determine the types $\vct t_{-ij}(\omega_{-e})$ of all agents other than buyer $i$ and seller $j$.
Consider a type profile $\vct t=(v_i,c_j,\vct t_{-ij}(\omega_{-e}))$ such that $e=(i,j)\in M^*(\vct t)$.
Increase buyer $i$'s value by $\delta\ge0$, holding all other types fixed.
Every matching containing buyer $i$ gains exactly $\delta$ in objective value, while every matching omitting $i$ is unchanged.
Relative comparisons among matchings that contain $i$ therefore do not change, and a matching that omits $i$ cannot overtake the selected matching.
The fixed tie-breaking order selects the same matching.
Lowering seller $j$'s cost is symmetric.
In quantile coordinates, these two improvements are precisely $\rho'\le \rho$ and $\sigma'\le \sigma$, proving the down-set property.

Every edge in a first-best matching contributes nonnegative GFT.
If a selected edge had $v_i-c_j<0$, deleting that edge would preserve feasibility by downward closedness and strictly increase GFT.
This proves (ii).
Measurability follows because $\cF$ is finite and the tie-broken maximizer is determined by finitely many comparisons of measurable functions.

Finally,
\[
  \GFT(M^*;\vct t)
  =\sum_{e=(i,j)\in E}(v_i-c_j)\1\{e\in M^*\}.
\]
Represent all types by independent quantiles.
In the term for edge $e$, condition on every quantile except those of its two endpoints.
The remaining indicator is exactly $\1_{\cD_e^{\omega_{-e}}}$.
Summing over edges and applying Tonelli proves (iii); nonnegativity from (ii) justifies the interchange.
\end{proof}

The same monotonicity gives a second property that will be used when we assemble the local posted prices.

\begin{corollary}[Stable partner]\label{cor:stable-partner}
Fix all reports except buyer $i$'s report.
As buyer $i$ varies her report, all first-best matchings in which she is matched use the same seller.
The symmetric statement holds for every seller.
\end{corollary}

\begin{proof}
Take two reports at which buyer $i$ is matched, and let the lower report be the first one.
Increasing the report from the lower to the higher value preserves the entire tie-broken first-best matching by the argument in \cref{prop:edge-region}.
Hence the buyer has the same partner at both reports.
The seller statement is symmetric.
\end{proof}

\begin{remark}[Why the region is two-dimensional]\label{rem:nonrectangle}
Even after all other types are fixed, $\cD_e^{\omega_{-e}}$ is generally not a rectangle.
If the best feasible alternative to using edge $e$ generates GFT $w$, then selection of $e$ may require $Q_i(\rho)-Q_j(\sigma)\ge w$.
The buyer threshold then depends on the seller quantile and vice versa.
The local inequality below must account for this dependence.
\end{remark}

\subsection{A local bound on GFT}\label{subsec:local-problem}

A measurable set $\cD\subseteq(0,1)^2$ is a down-set if $(\rho,\sigma)\in \cD$ and $0<\rho'\le \rho$, $0<\sigma'\le \sigma$ imply $(\rho',\sigma')\in \cD$.
Define the horizontal and vertical section lengths
\begin{align*}
  \ell_B(\sigma)&:=\sup\{\rho\in(0,1):(\rho,\sigma)\in\cD\},\\
  \ell_S(\rho)&:=\sup\{\sigma\in(0,1):(\rho,\sigma)\in\cD\},
\end{align*}
with the supremum of an empty section interpreted as zero.
Since $\cD$ is a down-set, each section is an initial interval, with its upper endpoint possibly included:
\begin{equation*}
  \1_\cD(\rho,\sigma)
  =\1\{0<\rho<\ell_B(\sigma)\}
  =\1\{0<\sigma<\ell_S(\rho)\}
\end{equation*}
for almost every $(\rho,\sigma)$.
Both $\ell_B$ and $\ell_S$ are measurable and nonincreasing.
Their dependence on the fixed region $\cD$ is suppressed in the notation.

Let $Q_B$ be a nonincreasing integrable buyer quantile function and $Q_S$ a nondecreasing integrable seller quantile function.
Assume that trade generates nonnegative GFT on $\cD$:
\begin{equation*}
  Q_B(\rho)\ge Q_S(\sigma)
  \qquad\text{for almost every }(\rho,\sigma)\in \cD.
\end{equation*}

Fix $\sigma$.
A threshold serving the top $x$ buyer mass remains inside the horizontal section whenever $x<\ell_B(\sigma)$, and then earns expected profit $x(Q_B(x)-Q_S(\sigma))$.
We therefore define the horizontal posted-price benchmark
\begin{equation*}
  P_S(\cD)
  :=\int_0^1
      \sup_{0\le x<\ell_B(\sigma)}
      x\bigl(Q_B(x)-Q_S(\sigma)\bigr)\dd \sigma.
\end{equation*}
Symmetrically, the vertical procurement benchmark is
\begin{equation*}
  P_B(\cD)
  :=\int_0^1
      \sup_{0\le y<\ell_S(\rho)}
      y\bigl(Q_B(\rho)-Q_S(y)\bigr)\dd \rho.
\end{equation*}
Not trading is always available with zero profit, and empty sections contribute zero, so both quantities are nonnegative.
They are finite; for example, $xQ_B(x)\le\int_0^xQ_B(s)\dd s\le\E[v]$.

These benchmarks also have a geometric interpretation.
For fixed $\sigma$, plot the GFT $Q_B(\rho)-Q_S(\sigma)$ against the buyer rank $\rho$ on the horizontal section.
The profit from threshold $x$ is the area of a rectangle of width $x$ and height $Q_B(x)-Q_S(\sigma)$ beneath this curve, with its lower-left corner at the origin.
Averaging the supremum of these rectangle areas over $\sigma$ gives $P_S(\cD)$; averaging the full areas under the curves gives the contribution of $\cD$ to first-best GFT.

The following example illustrates the two benchmarks; Figure~\ref{fig:posted-price-profit} depicts the seller-side profit on one horizontal section.
If $Q_B(\rho)=1-\rho$, $Q_S(\sigma)=\sigma$, and $\cD=\{(\rho,\sigma):\rho+\sigma<1\}$, then the contribution of $\cD$ to first-best GFT is $1/6$.
For a horizontal section at $\sigma$, the best threshold solves $\max_{0\le x\le1-\sigma}x(1-\sigma-x)$ and has value $(1-\sigma)^2/4$, so $P_S(\cD)=1/12$.
By symmetry $P_B(\cD)=1/12$.
Thus the two section benchmarks sum to first-best GFT in this triangular example.
It is not tight; \cref{sec:tight} gives a family for which the ratio approaches $e/2$.

\begin{figure}[htbp]
\centering
\begin{tikzpicture}[x=1.15cm,y=0.95cm,>=Latex,font=\small]
  \path[fill=black!8] (0,0)--(5,0)--(0,3)--cycle;
  \path[fill=black!23,draw=black!70,thick] (0,0) rectangle (2.5,1.5);
  \draw[thick] (0,3)--(5,0);
  \draw[->,thick] (0,0)--(5.65,0) node[right] {$\rho$};
  \draw[->,thick] (0,0)--(0,3.45) node[above] {GFT};
  \node[below left] at (0,0) {$0$};
  \draw (2.5,0)--(2.5,-0.08) node[below] {$x$};
  \draw (5,0)--(5,-0.08) node[below] {$\ell_B(\sigma)=1-\sigma$};
  \draw (0,1.5)--(-0.08,1.5) node[left] {$Q_B(x)-Q_S(\sigma)$};
  \draw (0,3)--(-0.08,3) node[left] {$1-\sigma$};
  \fill (2.5,1.5) circle (1.5pt);
  \node[align=center] at (1.25,0.75) {Seller-side\\profit};
  \node[above right] at (3.4,0.96) {$Q_B(\rho)-Q_S(\sigma)$};
\end{tikzpicture}
\caption{Posted-price profit at a fixed seller rank $\sigma$ in the uniform example.
The largest rectangle has width $x=(1-\sigma)/2$ and height $Q_B(x)-Q_S(\sigma)$; its area is the optimal seller-side profit on this section.
The full shaded triangle represents the contribution of this section to first-best GFT.
Averaging the optimal rectangle areas over $\sigma$ gives $P_S(\cD)$.}
\label{fig:posted-price-profit}
\end{figure}
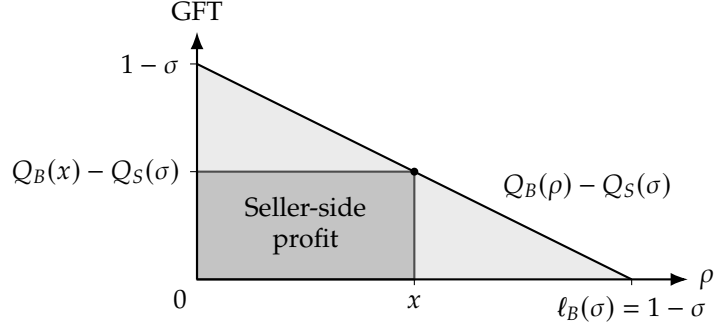

\begin{theorem}[Local bound on GFT]\label{thm:local}
For every $Q_B,Q_S,\cD$ satisfying the preceding assumptions,
\begin{equation}\label{eq:local-ineq}
  \int_\cD\bigl(Q_B(\rho)-Q_S(\sigma)\bigr)\dd \rho\dd \sigma
  \le
  \frac e2\bigl(P_S(\cD)+P_B(\cD)\bigr).
\end{equation}
\end{theorem}

Applying \cref{thm:local} to each first-best edge-selection region yields
\begin{equation}\label{eq:after-local}
  \FB
  \le\frac e2
  \sum_{e\in E}\E_{\omega_{-e}}
  \left[P_S(\cD_e^{\omega_{-e}})+P_B(\cD_e^{\omega_{-e}})\right].
\end{equation}
We defer the proof of \cref{thm:local} to \cref{subsec:local-proof}.
The next subsection shows that the two sums of local benchmarks in \eqref{eq:after-local} are bounded by $\OPTS$ and $\OPTB$, completing the proof of \cref{thm:main}.

\subsection{From local bounds to the approximation guarantee}\label{subsec:implementation}

Equation~\eqref{eq:after-local} has reduced the market theorem to an implementation question.
The horizontal benchmark is defined edge by edge and seller-section by seller-section; a priori, running all of those posted prices at once could violate matching or feasibility constraints.
The stable-partner property in \cref{cor:stable-partner} prevents this conflict.

\begin{proposition}[Realizing the section benchmarks]\label{prop:realize-benchmarks}
The local benchmarks satisfy
\begin{equation*}
  \sum_{e\in E}\E_{\omega_{-e}}[P_S(\cD_e^{\omega_{-e}})]\le\OPTS,
  \qquad
  \sum_{e\in E}\E_{\omega_{-e}}[P_B(\cD_e^{\omega_{-e}})]\le\OPTB.
\end{equation*}
\end{proposition}

\begin{proof}
We prove the first inequality; the second is symmetric.
If $E=\varnothing$, there is nothing to prove.
Fix $\varepsilon>0$ and set $\eta:=\varepsilon/|E|$.

We first choose a threshold for each potential edge.
Fix an edge $e=(i,j)$ and fix all reports except buyer $i$'s report; denote this fixed profile by $\vct t_{-i}$.
Let
\begin{equation*}
  \ell_{B,e}(\vct t_{-i})
  :=\sup\{\rho\in(0,1):e\in M^*(Q_i(\rho),\vct t_{-i})\},
\end{equation*}
with value zero if the set is empty.
By \cref{prop:edge-region}, this is the length of an initial buyer-quantile interval.
If $\ell_{B,e}(\vct t_{-i})=0$, set $x_e(\vct t_{-i})=0$.
Otherwise choose a measurable threshold $x_e(\vct t_{-i})<\ell_{B,e}(\vct t_{-i})$ such that
\begin{equation}\label{eq:near-opt-threshold}
  x_e\bigl(Q_i(x_e)-c_j\bigr)
  \ge
  \sup_{0\le x<\ell_{B,e}(\vct t_{-i})}
  x\bigl(Q_i(x)-c_j\bigr)-\eta.
\end{equation}
Appendix~\ref{app:measurable-posted} gives a measurable choice.
Offer buyer $i$ the threshold price corresponding to the top $x_e$ quantile mass on edge $e$.

For fixed $\vct t_{-i}$, \cref{cor:stable-partner} implies that at most one edge incident to buyer $i$ has $\ell_{B,e}(\vct t_{-i})>0$.
Hence these rules give a single threshold rule for buyer $i$.
In particular, the threshold is independent of buyer $i$'s report, so the usual threshold-payment argument gives DSIC and ex-post IR for the buyers.

We next verify that the accepted edges form a feasible matching.
Suppose buyer $i$ accepts the offer on edge $e$.
Away from a boundary atom, her realized quantile satisfies
\[
  \rho_i<x_e(\vct t_{-i})<\ell_{B,e}(\vct t_{-i}),
\]
so $e$ belongs to the tie-broken first-best matching at the realized profile.
At a boundary atom, the same conclusion follows from \cref{lem:atoms-not-split}.
Thus every accepted edge belongs to the same matching $M^*(\vct t)$.
The set of accepted edges is a subset of that matching and is therefore feasible by downward closedness.

It remains to bound the expected seller-side profit.
Conditional on $\omega_{-e}$ and seller $j$'s quantile $\sigma$, edge $e$ contributes
\[
  x_e\bigl(Q_i(x_e)-Q_j(\sigma)\bigr)
\]
to expected seller-side profit.
By \eqref{eq:near-opt-threshold}, this is within $\eta$ of the corresponding horizontal-section benchmark.
Integrating over $\sigma$ and $\omega_{-e}$ and then summing over edges gives a feasible buyer auction with expected objective at least
\[
  \sum_{e\in E}\E_{\omega_{-e}}[P_S(\cD_e^{\omega_{-e}})]-\varepsilon.
\]
Since $\varepsilon$ is arbitrary, the first inequality of the proposition follows.

For the second inequality, fix all reports except one seller's report, choose a near-optimal bottom-quantile threshold on the corresponding vertical section, and repeat the same argument.
Stable partner gives at most one active procurement offer per seller, and every accepted edge is again contained in the realized first-best matching.
\end{proof}

\begin{proof}[Proof of \cref{thm:main}]
Combining \eqref{eq:after-local} with \cref{prop:realize-benchmarks} gives
\[
  \FB\le\frac e2(\OPTS+\OPTB).
\]
By \cref{prop:gsom-gbom}, this is exactly \eqref{eq:main-profit}.
By \cref{lem:gft-profit},
\begin{align*}
  \E[\GFT(\mathcal M)]
  &=\frac12\E[\GFT(\GSOM)]+\frac12\E[\GFT(\GBOM)]\\
  &\ge\frac12\Pi_S(\GSOM)+\frac12\Pi_B(\GBOM)\\
  &\ge\frac1e\FB.
\end{align*}
The random choice defining $\mathcal M$ is independent of the reports.
Since both components are DSIC, ex-post IR, and ex-ante WBB, $\mathcal M$ inherits the same properties.
\end{proof}

\begin{remark}[Where the assumptions are used]\label{rem:assumptions}
Independence gives the product-quantile representation and permits conditioning on one edge while leaving its two endpoint quantiles independent.
The matching constraint gives the stable-partner property.
Downward closedness is used twice: a first-best matching cannot contain an edge that contributes negative GFT, and any subset of a feasible first-best matching remains feasible in the posted-price construction.
No matroid or exchange property is used.
\end{remark}

\subsection{Proof of the local bound on GFT}\label{subsec:local-proof}

Fix $Q_B,Q_S$, and $\cD$ as in \cref{thm:local}.
On a randomly contracted copy of $\cD$, we choose a posted price for each horizontal section and for each vertical section.
We evaluate the buyer and seller contributions separately to compare the expected profit of these prices with the GFT on $\cD$.

To construct the price rules, draw $h\sim\operatorname{Unif}(0,1)$ independently of the endpoint ranks, and set
\begin{equation*}
  \theta:=e^{-h},\qquad \eta:=e^{-(1-h)}.
\end{equation*}
For each realization of $h$, define
\[
  \cD_h:=\{(\theta\rho,\eta\sigma):(\rho,\sigma)\in\cD\}.
\]
As $h$ increases, $\theta$ decreases and $\eta$ increases, while $\theta\eta=e^{-1}$.
Since $0<\theta,\eta<1$ and $\cD$ is a down-set, we have $\cD_h\subseteq\cD$.

Consider two posted-price rules on $\cD_h$.
On a horizontal section at seller rank $\sigma\in(0,\eta)$, serve the top $x_h(\sigma)$ buyer mass at price $Q_B(x_h(\sigma))$; on a vertical section at buyer rank $\rho\in(0,\theta)$, procure from the bottom $y_h(\rho)$ seller mass at price $Q_S(y_h(\rho))$, where
\[
  x_h(\sigma):=\theta\ell_B(\sigma/\eta),
  \qquad
  y_h(\rho):=\eta\ell_S(\rho/\theta).
\]
Each rule makes no offer outside the indicated range or on an empty section.
Products involving a zero threshold are interpreted as zero, so no quantile is evaluated at zero.
Both rules trade on $\cD_h$ in quantile space, up to null boundaries, with boundary randomization at atoms as in Appendix~\ref{app:quantile}.

The thresholds are admissible for the original sections.
Indeed, whenever $\ell_B(\sigma/\eta)>0$, monotonicity of $\ell_B$ gives
\[
  0<x_h(\sigma)<\ell_B(\sigma/\eta)\le\ell_B(\sigma),
\]
and the vertical statement is symmetric.
Consequently, their one-sided profits satisfy
\begin{align}
  \pi_S(h)
  &:=\int_0^{\eta}x_h(\sigma)
       \bigl(Q_B(x_h(\sigma))-Q_S(\sigma)\bigr)\dd\sigma
    \le P_S(\cD),
    \label{eq:contracted-seller-profit}\\
  \pi_B(h)
  &:=\int_0^{\theta}y_h(\rho)
       \bigl(Q_B(\rho)-Q_S(y_h(\rho))\bigr)\dd\rho
    \le P_B(\cD).
    \label{eq:contracted-buyer-profit}
\end{align}

Conditional on $h$, choose the horizontal rule with probability $1-h$ and the vertical rule with probability $h$.
We will show that the resulting expected local one-sided profit satisfies the exact identity
\begin{equation}\label{eq:random-price-identity}
  \E_h\bigl[(1-h)\pi_S(h)+h\pi_B(h)\bigr]
  =\frac1e\int_\cD
      \bigl(Q_B(\rho)-Q_S(\sigma)\bigr)\dd\rho\dd\sigma.
\end{equation}
The desired bound then follows from \eqref{eq:contracted-seller-profit} and \eqref{eq:contracted-buyer-profit}, since $\E[h]=\E[1-h]=1/2$.

To evaluate the expected profit of these price rules, we first prove the following integral identity.

\begin{lemma}[Exponential scale identity]\label{lem:exponential-scale}
Let $Q:(0,1)\to[0,\infty)$ be integrable.
For every $t\in[0,1]$,
\begin{equation}\label{eq:exponential-scale-identity}
  \E_h\!\left[
    \eta\left(h\int_0^{\theta t}Q(u)\dd u+(1-h)\theta tQ(\theta t)\right)
  \right]
  =\frac1e\int_0^tQ(u)\dd u,
\end{equation}
where the integrand is defined to be zero when $t=0$.
\end{lemma}

\begin{proof}
Recall that $\theta=e^{-h}$ and $\eta=e^{h-1}$.
The case $t=0$ is immediate.
For $t>0$, the map $h\mapsto e^{-h}t$ is smooth and bi-Lipschitz from $[0,1]$ onto $[t/e,t]$.
Since $Q$ is integrable, the function $h\mapsto(1-h)e^{h-1}\int_0^{e^{-h}t}Q(u)\dd u$ is absolutely continuous.
By the product and chain rules,
\[
\begin{aligned}
  &\frac{\dd}{\dd h}\left[(1-h)e^{h-1}\int_0^{e^{-h}t}Q(u)\dd u\right]\\
  &\quad=-h e^{h-1}\int_0^{e^{-h}t}Q(u)\dd u
       -(1-h)e^{h-1}e^{-h}tQ(e^{-h}t)\\
  &\quad=-e^{h-1}\left[h\int_0^{e^{-h}t}Q(u)\dd u
       +(1-h)e^{-h}tQ(e^{-h}t)\right]
\end{aligned}
\]
for almost every $h\in(0,1)$.
Since $h\sim\operatorname{Unif}(0,1)$, integrating the preceding identity gives
\begin{align*}
  &\E_h\!\left[
    \eta\left(h\int_0^{\theta t}Q(u)\dd u
    +(1-h)\theta tQ(\theta t)\right)
  \right]\\
  &\quad=-\left[
    (1-h)e^{h-1}\int_0^{e^{-h}t}Q(u)\dd u
  \right]_{h=0}^{h=1}
  =\frac1e\int_0^tQ(u)\dd u.
\end{align*}
This proves \eqref{eq:exponential-scale-identity}.
\end{proof}

\begin{proof}[Proof of \cref{thm:local}]
We prove \eqref{eq:random-price-identity} by separating the terms involving $Q_B$ from those involving $Q_S$ in $(1-h)\pi_S(h)+h\pi_B(h)$.
Both price rules trade on $\cD_h$.
Using horizontal sections of $\cD_h$ and $x_h(\eta\sigma)=\theta\ell_B(\sigma)$, we can write the buyer-value terms as
\begin{align*}
  &(1-h)\int_0^{\eta}x_h(\sigma)Q_B(x_h(\sigma))\dd\sigma
    +h\int_{\cD_h}Q_B(\rho)\dd\rho\dd\sigma\\
  &\quad=\eta\int_0^1
    \Bigl[(1-h)\theta\ell_B(\sigma)Q_B(\theta\ell_B(\sigma))
      +h\int_0^{\theta\ell_B(\sigma)}Q_B(\rho)\dd\rho\Bigr]\dd\sigma.
\end{align*}
Applying \cref{lem:exponential-scale} with $Q=Q_B$ and $t=\ell_B(\sigma)$, together with Tonelli's theorem, shows that the expectation of this expression is
\[
  \frac1e\int_0^1\int_0^{\ell_B(\sigma)}Q_B(\rho)\dd\rho\dd\sigma
  =\frac1e\int_\cD Q_B(\rho)\dd\rho\dd\sigma.
\]

Similarly, using vertical sections of $\cD_h$ and $y_h(\theta\rho)=\eta\ell_S(\rho)$, we can write the seller-cost terms as
\begin{align*}
  &(1-h)\int_{\cD_h}Q_S(\sigma)\dd\rho\dd\sigma
    +h\int_0^{\theta}y_h(\rho)Q_S(y_h(\rho))\dd\rho\\
  &\quad=\theta\int_0^1
    \Bigl[(1-h)\int_0^{\eta\ell_S(\rho)}Q_S(\sigma)\dd\sigma
      +h\eta\ell_S(\rho)Q_S(\eta\ell_S(\rho))\Bigr]\dd\rho.
\end{align*}
Replacing $h$ by $1-h$ preserves its distribution and interchanges $\theta$ and $\eta$.
Applying \cref{lem:exponential-scale} with $Q=Q_S$ and $t=\ell_S(\rho)$, together with Tonelli's theorem, therefore shows that the expectation of the seller-cost terms is
\[
  \frac1e\int_0^1\int_0^{\ell_S(\rho)}Q_S(\sigma)\dd\sigma\dd\rho
  =\frac1e\int_\cD Q_S(\sigma)\dd\rho\dd\sigma.
\]
Both expectations are finite because $Q_B$ and $Q_S$ are integrable.
The weighted profit is the buyer-value terms minus the seller-cost terms, so subtracting these expectations proves \eqref{eq:random-price-identity}.
Finally, \eqref{eq:contracted-seller-profit} and \eqref{eq:contracted-buyer-profit} imply
\[
  \frac1e\int_\cD
    \bigl(Q_B(\rho)-Q_S(\sigma)\bigr)\dd\rho\dd\sigma
  \le\E_h\bigl[(1-h)P_S(\cD)+hP_B(\cD)\bigr]
  =\frac12\bigl(P_S(\cD)+P_B(\cD)\bigr),
\]
which proves \eqref{eq:local-ineq}.
\end{proof}

\section{Tight \texorpdfstring{$e$}{e}-Approximation Upper Bound}\label{sec:tight}

We now show that the $e$-approximation for GRO cannot be improved over downward-closed matching markets.
Our goal is to construct markets in which GRO's expected GFT approaches a $1/e$ fraction of first-best GFT. The key challenge is to turn tightness of the local profit inequality into tightness for the mechanism's actual GFT.
We achieve this using regular type distributions supported on $[0,1]$, pairwise disjoint trading edges, and a single knapsack constraint.

A buyer distribution is \emph{regular} if its revenue curve $R$ is concave after linear interpolation across the quantile intervals corresponding to point masses; a seller distribution is regular if its similarly interpolated cost curve $K$ is convex.
Appendix~\ref{app:quantile} specifies the interpolation convention.
In our construction, $R$ is already concave and $K$ is already convex, so interpolation leaves both curves unchanged.

\begin{theorem}[Tightness]\label{thm:tightness}
For every $\varepsilon>0$, there is a finite market satisfying the assumptions of \cref{thm:main} such that
\begin{enumerate}[label=(\roman*)]
  \item every type distribution is regular and supported on $[0,1]$;
  \item all trading edges have pairwise disjoint endpoints;
  \item feasibility is described by a single knapsack constraint; and
  \item for $\mathcal M$,
  \begin{equation*}
    \frac{\FB}{\E[\GFT(\mathcal M)]}\ge e-\varepsilon,
  \end{equation*}
  and
  \begin{equation*}
    \frac{\FB}
    {\frac12\left(\Pi_S(\GSOM)+\Pi_B(\GBOM)\right)}
    \ge e-\varepsilon.
  \end{equation*}
\end{enumerate}
\end{theorem}

The proof of \cref{thm:tightness} has two layers.
We first construct a bounded regular bilateral instance for which the local inequality is asymptotically tight.
This controls the one-sided profit benchmarks.
We then add a deterministic alternative under one knapsack constraint so that both components of GRO exhibit the same asymptotic gap in actual GFT.

\subsection{A bilateral building block}\label{subsec:bilateral-tight}

Fix $L\ge2$, let $p=e^{-L}$, and let $X,Y$ be independent $\operatorname{Exp}(1)$ random variables.
Define
\begin{equation*}
  v:=\min\{X,L\},
  \qquad
  c:=L-\min\{Y,L\}.
\end{equation*}
The buyer distribution is exponential on $[0,L)$ with an atom of mass $p$ at $L$; the seller distribution is its mirror image, with an atom of mass $p$ at zero.
The exponential form is useful because, away from the caps, the buyer virtual value is exactly one below value and the seller virtual cost exactly one above cost.
The caps keep the support bounded while preserving regularity.

\begin{lemma}[Capped-exponential calculations]\label{lem:capped-exp}
The buyer and seller distributions defined above are regular.
Their ironed virtual value and virtual cost admit the type-space representatives
\begin{equation}\label{eq:virtual-capped}
  \phi(v)=
  \begin{cases}
    v-1,&v<L,\\
    L,&v=L,
  \end{cases}
  \qquad
  \psi(c)=
  \begin{cases}
    0,&c=0,\\
    c+1,&c>0.
  \end{cases}
\end{equation}
They satisfy
\begin{equation}\label{eq:mirror-virtual}
  \psi(L-z)=L-\phi(z),\qquad 0\le z\le L.
\end{equation}
Moreover, with
\begin{align*}
  g_L&:=\E[\pos{v-c}],\\
  p_L^S&:=\E[\pos{\phi(v)-c}],\\
  p_L^B&:=\E[\pos{v-\psi(c)}],
\end{align*}
we have
\begin{equation}\label{eq:capped-values}
  g_L=Le^{-L},
  \qquad
  p_L^S=p_L^B=e^{-L}\left(1+\frac{L-2}{e}\right).
\end{equation}
\end{lemma}

\begin{proof}
The buyer upper-quantile function and seller lower-quantile function are
\[
  Q_B(\rho)=
  \begin{cases}
    L,&0<\rho\le p,\\
    -\log \rho,&p<\rho<1,
  \end{cases}
  \qquad
  Q_S(\sigma)=
  \begin{cases}
    0,&0<\sigma\le p,\\
    L+\log \sigma,&p<\sigma<1.
  \end{cases}
\]
Hence
\[
  R(\rho)=
  \begin{cases}
    L\rho,&\rho\le p,\\
    -\rho\log \rho,&\rho>p,
  \end{cases}
  \qquad
  K(\sigma)=
  \begin{cases}
    0,&\sigma\le p,\\
    \sigma(L+\log \sigma),&\sigma>p.
  \end{cases}
\]
The slope of $R$ is $L$ on $(0,p)$ and $-\log \rho-1$ on $(p,1)$, with a downward jump at $p$.
Thus $R$ is concave.
The slope of $K$ is zero on $(0,p)$ and $L+\log \sigma+1$ on $(p,1)$, with an upward jump at $p$.
Thus $K$ is convex.
No further ironing is needed, and translating the slopes back to type space gives \eqref{eq:virtual-capped}.
The mirror identity \eqref{eq:mirror-virtual} follows directly.

Let $Z:=\min\{Y,L\}$, so $c=L-Z$ and $v,Z$ are independent with the same law.
For a fixed buyer value $0\le v<L$,
\begin{equation*}
  \E_Z[\pos{v+Z-L}]
  =\int_{L-v}^L(v+z-L)e^{-z}\dd z+pv
  =p(e^v-1).
\end{equation*}
At the atom $v=L$, the conditional expected GFT is $\E[Z]=1-p$.
Therefore
\[
  g_L
  =\int_0^Lp(e^v-1)e^{-v}\dd v+p(1-p)
  =Lp.
\]

For seller-side virtual surplus, the expression $\pos{v-1+Z-L}$ vanishes when $v\le1$.
For a fixed buyer value $1<v<L$,
\begin{equation*}
  \E_Z[\pos{v-1+Z-L}]=p(e^{v-1}-1).
\end{equation*}
At $v=L$, the virtual value is $L$, so the conditional expectation is again $1-p$.
Hence
\[
  p_L^S
  =\int_1^Lp(e^{v-1}-1)e^{-v}\dd v+p(1-p)
  =p\left(1+\frac{L-2}{e}\right).
\]
Finally, \eqref{eq:mirror-virtual} and exchangeability of $v$ and $Z$ give
\[
  p_L^B
  =\E[\pos{v-L+\phi(Z)}]
  =\E[\pos{\phi(v)+Z-L}]
  =p_L^S.
\]
\end{proof}

Let $\cD=\{(\rho,\sigma):Q_B(\rho)\ge Q_S(\sigma)\}$ be the efficient region of this bilateral instance.
Since $R$ is concave and $K$ is convex, the posted-price profit functions $R(x)-cx$ and $vy-K(y)$ are concave for every fixed $c$ and $v$, respectively.
Both functions vanish at zero, so each maximum equals the integral of the positive part of its derivative.
Averaging over the fixed endpoint's type therefore gives the unrestricted benchmarks $p_L^S$ and $p_L^B$.
Restricting thresholds to their sections of $\cD$ leaves these optima unchanged: any threshold beyond its section yields nonpositive profit, and including the section endpoint does not change the supremum by \cref{lem:endpoint-sup}.
Thus
\[
  P_S(\cD)=p_L^S,
  \qquad
  P_B(\cD)=p_L^B.
\]
Consequently,
\begin{equation*}
  \frac{g_L}{p_L^S}
  =\frac{eL}{L+e-2}\longrightarrow e,
  \qquad
  \frac{g_L}{p_L^S+p_L^B}\longrightarrow\frac e2.
\end{equation*}
So the coefficient $e/2$ in the local inequality is already sharp for bounded regular bilateral distributions.

\subsection{A tight example with a knapsack constraint}\label{subsec:knapsack}

The calculation above compares first-best GFT with one-sided \emph{profit}.
It does not yet compare first-best GFT with the actual GFT produced by either component of GRO.
The difference is the utility of the strategic side.
A one-sided auction can have small profit while still generating substantial GFT because buyers or sellers obtain large information rents.

We place many independent copies of the bilateral instance next to one deterministic edge.
The GFT from the deterministic edge is chosen to lie strictly above the expected maximum virtual surplus of the random edges for either component of GRO and strictly below their expected maximum GFT.
A single knapsack constraint then forces first best and the two components of GRO to choose different alternatives with high probability.

Set
\begin{equation*}
  d_L:=\left(1+\frac1L\right)p_L^S.
\end{equation*}
For every $L\ge2$,
\begin{equation*}
  p_L^S=p_L^B<d_L<g_L.
\end{equation*}
Only the last inequality is nontrivial.
By \eqref{eq:capped-values}, it is equivalent to
\[
  eL^2>(L+1)(L+e-2),
\]
and the difference between the two sides is $(e-1)L(L-1)-(e-2)>0$.

Fix $N$.
Take $N$ independent capped-exponential edges $e_1,\ldots,e_N$, all with disjoint endpoints.
Add one disjoint edge $e_0$ with deterministic buyer value $Nd_L$ and deterministic seller cost zero.
Define
\begin{equation*}
  \cF_{L,N}:=2^{\{e_1,\dots,e_N\}}\cup\{\{e_0\}\}.
\end{equation*}
Equivalently, give each random edge weight one, give $e_0$ weight $N$, and impose knapsack capacity $N$.
Thus every subset of the random edges is feasible, while $e_0$ must be chosen alone.

For $i=1,\ldots,N$, write
\begin{equation*}
  G_i:=\pos{v_i-c_i},
  \qquad
  P_i^S:=\pos{\phi(v_i)-c_i},
  \qquad
  P_i^B:=\pos{v_i-\psi(c_i)}.
\end{equation*}
These random variables lie in $[0,L]$ and have means $g_L,p_L^S,p_L^B$.
Among subsets of the random edges, the maximum GFT, seller-side virtual surplus, and buyer-side virtual surplus are $\sum_iG_i$, $\sum_iP_i^S$, and $\sum_iP_i^B$, respectively.
The deterministic edge contributes $Nd_L$ to GFT and to both virtual-surplus objectives.

Write $\FB_{L,N}$ for first-best GFT and $\GSOM_{L,N},\GBOM_{L,N}$ for the two mechanisms in this market.

\begin{proposition}[Asymptotic behavior]\label{prop:asymptotic}
For every fixed $L\ge2$, as $N\to\infty$,
\begin{align*}
  \frac{\FB_{L,N}}N&\longrightarrow g_L,\\
  \frac{\E[\GFT(\GSOM_{L,N})]}N&\longrightarrow d_L,
  \qquad
  \frac{\E[\GFT(\GBOM_{L,N})]}N\longrightarrow d_L,\\
  \frac{\Pi_S(\GSOM_{L,N})}N&\longrightarrow d_L,
  \qquad
  \frac{\Pi_B(\GBOM_{L,N})}N\longrightarrow d_L.
\end{align*}
\end{proposition}

\begin{proof}
Because $G_i$, $P_i^S$, and $P_i^B$ are bounded, the law of large numbers gives almost-sure and $L^1$ convergence:
\[
  \frac1N\sum_{i=1}^NG_i\to g_L,
  \qquad
  \frac1N\sum_{i=1}^NP_i^S\to p_L^S,
  \qquad
  \frac1N\sum_{i=1}^NP_i^B\to p_L^B.
\]

The first-best rule compares the best subset of random edges with the deterministic edge.
Hence
\[
  \frac{\FB_{L,N}}N
  =\E\left[\max\left\{\frac1N\sum_{i=1}^NG_i,d_L\right\}\right].
\]
Since $g_L>d_L$ and $x\mapsto\max\{x,d_L\}$ is one-Lipschitz, $L^1$ convergence gives the claimed limit for $\FB_{L,N}/N$.

For GSOM, let
\[
  \mathcal E_N^S
  :=\left\{\frac1N\sum_{i=1}^NP_i^S\ge d_L\right\}.
\]
Since $p_L^S<d_L$, we have $\Prb(\mathcal E_N^S)\to0$.
Outside $\mathcal E_N^S$, the deterministic edge is the unique seller-side virtual-surplus maximizer, so GSOM selects it and obtains GFT $Nd_L$.
On $\mathcal E_N^S$, every random edge selected by GSOM contributes nonnegative GFT because $\phi(v)\le v$, and the realized GFT divided by $N$ lies in $[0,L]$.
Therefore
\[
  \left|\frac{\E[\GFT(\GSOM_{L,N})]}N-d_L\right|
  \le L\Prb(\mathcal E_N^S)\longrightarrow0.
\]
The GBOM limit is identical, using $\psi(c)\ge c$ and $p_L^B<d_L$.

Finally, we evaluate the one-sided profits.
The distributions are regular, and fixed tie-breaking makes the virtual-surplus maximizing allocation constant on each quantile interval corresponding to an atom.
The one-sided payment identities in Appendix~\ref{app:one-sided} therefore give
\begin{align*}
  \Pi_S(\GSOM_{L,N})
  &=\E\left[\max\left\{\sum_{i=1}^NP_i^S,Nd_L\right\}\right],\\
  \Pi_B(\GBOM_{L,N})
  &=\E\left[\max\left\{\sum_{i=1}^NP_i^B,Nd_L\right\}\right].
\end{align*}
The claimed one-sided profit limits follow again from $L^1$ convergence and the one-Lipschitz property of the maximum.
\end{proof}

Let $\mathcal M_{L,N}$ denote GRO instantiated in this market, i.e., the equal mixture of $\GSOM_{L,N}$ and $\GBOM_{L,N}$.
For fixed $L$, \cref{prop:asymptotic} gives
\begin{equation*}
  \lim_{N\to\infty}
  \frac{\FB_{L,N}}{\E[\GFT(\mathcal M_{L,N})]}
  =\frac{g_L}{d_L}
  =\frac{eL}{(1+1/L)(L+e-2)}.
\end{equation*}
The same limit holds if the denominator is $\frac12(\Pi_S(\GSOM_{L,N})+\Pi_B(\GBOM_{L,N}))$.
Letting $L\to\infty$ in this expression gives
\begin{equation*}
  \lim_{L\to\infty}\lim_{N\to\infty}
  \frac{\FB_{L,N}}{\E[\GFT(\mathcal M_{L,N})]}=e,
\end{equation*}
and the same limit for the profit ratio.

It remains only to normalize the supports.
For each finite pair $(L,N)$, divide all values and costs by
\[
  \Lambda_{L,N}:=\max\{L,Nd_L\}.
\]
Positive scaling multiplies GFT, payments, virtual values, and virtual costs by the same factor.
It preserves allocations, regularity, and all approximation ratios.
The deterministic buyer and seller distributions are regular because their revenue and cost curves are linear in quantile space.
After scaling, every type lies in $[0,1]$ and feasibility is still described by the same single knapsack constraint.
Choosing $L$ and then $N$ sufficiently large proves \cref{thm:tightness}.

\section{Conclusion}\label{sec:discussion}

We determine the exact worst-case approximation guarantee of GRO in downward-closed matching markets. GRO always obtains at least a \(1/e\) fraction of first-best GFT, and this guarantee is tight even with regular type distributions supported on \([0,1]\), pairwise disjoint trading edges, and a single knapsack constraint. As a consequence, the random-offerer mechanism in bilateral trade also obtains a \(1/e\) fraction of first-best GFT.

Several questions remain open. Our tight construction relies on a global knapsack alternative, and therefore does not settle the exact guarantee of GRO for ordinary bipartite matching without additional feasibility constraints. In bilateral trade, the exact worst-case guarantee of the random-offerer mechanism is also still unknown. Our analysis relies on the down-set structure of first-best edge-selection regions and the stable-partner property of single-parameter matching markets. It would be interesting to understand whether similar local arguments can be developed beyond this setting, particularly in multi-parameter two-sided markets.

\section*{Acknowledgements}
Generative AI tools, especially GPT 5.6 and GPT 6, assisted with many parts of the mathematical development and exposition.
These tools were also used for writing the early-stage manuscript of this paper.
The authors carefully verified the mathematical arguments and revised the exposition for clarity.
The authors take full responsibility for the final text and results.

\bibliographystyle{alpha}
\bibliography{references}

\clearpage
\appendix

\section{Quantiles, endpoints, and atoms}\label{app:quantile}

This appendix fixes the quantile conventions at atoms and endpoints and specifies the regularity convention used in \cref{sec:tight}.

For a nonnegative random variable $v$, choose a nonincreasing left-continuous upper-quantile representation $Q_B:(0,1)\to[0,\infty)$ such that $Q_B(\rho)$ has the law of $v$ for $\rho\sim\operatorname{Unif}(0,1)$.
For a nonnegative cost $c$, choose a nondecreasing left-continuous lower-quantile representation $Q_S$.
Generalized inverses provide such versions after changing only countably many points.

For $x\in(0,1)$, a top-$x$ threshold rule allocates to exactly the top $x$ probability mass.
Put $p=Q_B(x)$.
Then
\begin{equation*}
  \Prb[v>p]\le x\le\Prb[v\ge p].
\end{equation*}
If $\Prb[v=p]>0$, allocate to type $p$ with probability
\begin{equation*}
  \alpha_x
  :=\frac{x-\Prb[v>p]}{\Prb[v=p]}.
\end{equation*}
If the boundary has zero probability, no randomization is needed.
Charge $p$ conditional on allocation.
This rule is DSIC and IR, its allocation probability is $x$, and its expected payment is $xQ_B(x)=R(x)$.
At $x=0$, it makes no allocation and charges zero.
For bounded distributions, $x=1$ is defined using the finite left limit $Q_B(1):=\lim_{\rho\uparrow1}Q_B(\rho)$; every use in the main proof can equivalently be obtained as a limit from $x<1$.

If $x=x(z)$ is a measurable function of auxiliary reports $z$, define $p(z)=Q_B(x(z))$ on $\{0<x(z)<1\}$ and set $p(z)=0$ on $\{x(z)=0\}$.
The boundary-randomization maps are measurable, so allocation and payment are jointly measurable in $(v,z)$.
A bottom-$y$ procurement threshold is symmetric and has expected payment $yQ_S(y)=K(y)$.

\begin{lemma}[Section endpoints]\label{lem:endpoint-sup}
For $t\in(0,1)$ and $a,b\ge0$,
\begin{align*}
  \sup_{0\le x<t}x(Q_B(x)-a)
    &=\sup_{0\le x\le t}x(Q_B(x)-a),\\
  \sup_{0\le y<t}y(b-Q_S(y))
    &=\sup_{0\le y\le t}y(b-Q_S(y)).
\end{align*}
The same identities hold for $t=1$ when $Q_B,Q_S$ are bounded and extended to $1$ by their left limits.
\end{lemma}

\begin{proof}
Take $x_n\uparrow t$.
Monotonicity gives $Q_B(x_n)\ge Q_B(t)$, so
\begin{equation*}
  \liminf_{n\to\infty}x_n(Q_B(x_n)-a)
  \ge t(Q_B(t)-a).
\end{equation*}
Thus adding the endpoint cannot increase the first supremum.
The seller identity follows from $Q_S(y_n)\le Q_S(t)$ for $y_n\uparrow t$.
\end{proof}

For an atom at buyer value $a$, the corresponding quantile interval has endpoints $\Prb[v>a]$ and $\Prb[v\ge a]$; for an atom at seller cost $a$, its endpoints are $\Prb[c<a]$ and $\Prb[c\le a]$.
Let $R^{\mathrm{lin}}$ and $K^{\mathrm{lin}}$ be obtained from $R$ and $K$ by linear interpolation between the endpoint values on each such interval, leaving the curves unchanged elsewhere.
The regularity condition in \cref{sec:tight} is concavity of $R^{\mathrm{lin}}$ for buyers and convexity of $K^{\mathrm{lin}}$ for sellers.
This interpolation is needed because an atom may follow a gap in the type support, in which case the original curve need not be continuous at the beginning of the atom interval.
Degenerate distributions are regular because their curves are linear.
Applying this interpolation to the individual revenue and cost curves does not change the functions $\overline R_i$ and $\underline K_j$ from \cref{subsec:one-sided}.

\section{Measurability of the posted-price construction}\label{app:measurable-posted}

We give the measurable-selection details used in \cref{prop:realize-benchmarks}.
Fix an enumeration $\{d_k\}_{k\ge1}$ of $\mathbb Q\cap(0,1)$.

\begin{lemma}[Measurable section lengths]\label{lem:measurable-sections}
Fix an edge $e=(i,j)$.
Holding every report except buyer $i$'s fixed, the length $\ell_{B,e}$ of the buyer-quantile selection interval is a measurable function of the fixed reports.
The analogous seller-quantile length $\ell_{S,e}$ is measurable.
\end{lemma}

\begin{proof}
If $\vct t_{-i}$ denotes the fixed reports, then
\begin{equation*}
  \ell_{B,e}(\vct t_{-i})
  =\int_0^1\1\{e\in M^*(Q_i(\rho),\vct t_{-i})\}\dd \rho.
\end{equation*}
The indicator is jointly measurable in $(\rho,\vct t_{-i})$ because $M^*$ is determined by finitely many measurable objective comparisons and fixed tie-breaking.
Its integral is measurable.
The seller statement is symmetric.
\end{proof}

\begin{lemma}[Measurable approximate threshold]\label{lem:measurable-threshold}
Let $t(z)\in[0,1]$ and $a(z)\ge0$ be measurable, and let $Q_B$ be nonincreasing.
Define
\begin{equation*}
  F(z,x):=
  \begin{cases}
    0,&x=0,\\
    x(Q_B(x)-a(z)),&0<x<1.
  \end{cases}
\end{equation*}
For every $\eta>0$, there exists a measurable $x_\eta(z)\in[0,1)$ such that $x_\eta(z)=0$ when $t(z)=0$, and, when $t(z)>0$,
\begin{equation*}
  x_\eta(z)<t(z),
  \qquad
  F(z,x_\eta(z))
  \ge\sup_{0\le x<t(z)}F(z,x)-\eta.
\end{equation*}
The symmetric statement holds for the seller objective that is zero at $y=0$ and equals $y(b(z)-Q_S(y))$ for $y>0$, where $Q_S$ is nondecreasing.
\end{lemma}

\begin{proof}
Define $f_0(z):=0$.
For $k\ge1$, set
\begin{equation*}
  f_k(z):=
  \begin{cases}
    d_k(Q_B(d_k)-a(z)),&d_k<t(z),\\
    -\infty,&d_k\ge t(z).
  \end{cases}
\end{equation*}
Then $m(z):=\sup_{k\ge0}f_k(z)$ is measurable.
When $t(z)=0$, choose $x_\eta(z)=0$.
Otherwise choose the smallest $k\ge0$ with $f_k(z)\ge m(z)-\eta$, and set $x_\eta(z)=0$ if $k=0$ and $x_\eta(z)=d_k$ otherwise.
This choice is measurable and lies strictly below $t(z)$.

Rational thresholds have the same supremum: for any $x\in(0,t(z))$, choose rationals $d_{k_n}\uparrow x$.
Since $Q_B$ is nonincreasing,
\begin{equation*}
  \liminf_{n\to\infty}
  d_{k_n}(Q_B(d_{k_n})-a(z))
  \ge x(Q_B(x)-a(z)).
\end{equation*}
Including the zero option proves the claimed approximation bound.
The seller statement is symmetric.
\end{proof}

\begin{lemma}[Atoms are not split]\label{lem:atoms-not-split}
Fix every report except one endpoint of an edge $e$.
If two quantile ranks represent the same reported type, either both ranks select $e$ in the first-best matching or neither does.
Consequently, a nonempty initial selection interval cannot end in the interior of an atom interval.
\end{lemma}

\begin{proof}
The first-best matching is a deterministic function of the reported type profile and the fixed tie-breaking rule.
Equal endpoint types therefore induce the same matching.
Since the edge-selection set is an initial interval, its boundary cannot split a positive-length level set of the quantile function.
\end{proof}

Together with Appendix~\ref{app:quantile}, the lemmas justify the measurable posted-price construction in \cref{prop:realize-benchmarks}, including threshold atoms.

\section{One-sided optimality of GSOM and GBOM}\label{app:one-sided}

We recall the one-sided virtual-surplus characterization used in \cref{prop:gsom-gbom}.
The two-sided implementation properties invoked there are from~\cite{BrustleCaiWuZhao2017}.

A \emph{buyer ironing interval} is a connected component of $\{\rho\in(0,1):\overline R_i(\rho)>R_i(\rho)\}$.
A \emph{seller ironing interval} is a connected component of $\{\sigma\in(0,1):\underline K_j(\sigma)<K_j(\sigma)\}$.
The corresponding derivative is constant on each ironing interval and on the interior of each quantile interval corresponding to an atom.

\subsection{Seller-side auction}

For buyer $i$, let $\rho_i\sim\operatorname{Unif}(0,1)$ be her upper-tail rank, and use the revenue curve $R_i$, its concave majorant $\overline R_i$, and the nonincreasing left derivative $\Phi_i$ defined in \cref{subsec:one-sided}.
The type-space function $\overline\phi_i$ agrees almost surely with $\Phi_i$ under this quantile representation.

Fix the seller costs.
Write $\vct\rho=(\rho_k)_{k\in B}$ for the vector of independent buyer ranks, and express allocation rules in these coordinates.
For a BIC buyer allocation rule $q$, let
\begin{equation*}
  a_i(\rho):=\E_{\vct\rho_{-i}}[q_i(\rho,\vct\rho_{-i})]
\end{equation*}
be buyer $i$'s interim allocation probability.
BIC implies that $a_i$ is nonincreasing in the upper-tail rank.
The payment identity and Myerson's ironing inequality give
\begin{equation*}
  \E[p_i]
  \le\int_0^1\Phi_i(\rho)a_i(\rho)\dd \rho
  =\E[\Phi_i(\rho_i)q_i(\vct\rho)].
\end{equation*}
For an irregular distribution, equality is attained when the lowest buyer type's utility is normalized to zero and $a_i$ is constant on every interval where $\overline R_i>R_i$.
Summing the payment bounds over buyers and subtracting realized seller costs shows that every feasible BIC and interim-IR seller-side auction satisfies
\begin{equation*}
\begin{aligned}
  \E[\pi_S(M,\vct p)]
  &\le
  \E\left[
    \sum_{(i,j)\in M}\bigl(\Phi_i(\rho_i)-c_j\bigr)
  \right]\\
  &\le
  \E\left[
    \max_{M\in\cF}
    \sum_{(i,j)\in M}\bigl(\Phi_i(\rho_i)-c_j\bigr)
  \right].
\end{aligned}
\end{equation*}
The pointwise maximizer in the last expression is monotone in every buyer's rank.
On an interval where $\overline R_i>R_i$, the slope $\Phi_i$ is constant; with fixed tie-breaking, the selected matching is also constant on that interval.
The allocation is therefore compatible with ironing.
Normalized threshold payments make both inequalities above tight.
Hence the GSOM allocation is one-sided optimal.

\subsection{Buyer-side procurement auction}

For seller $j$, let $\sigma_j\sim\operatorname{Unif}(0,1)$ be her lower-tail rank, and use the cost curve $K_j$, its convex minorant $\underline K_j$, and the nondecreasing left derivative $\Psi_j$ defined in \cref{subsec:one-sided}.
The type-space function $\overline\psi_j$ agrees almost surely with $\Psi_j$ under this quantile representation.
Fix the buyer values.
Write $\vct\sigma=(\sigma_k)_{k\in S}$ for the vector of independent seller ranks, and express allocation rules in these coordinates.
For a BIC procurement allocation rule $q$, let
\begin{equation*}
  a_j(\sigma):=\E_{\vct\sigma_{-j}}[q_j(\sigma,\vct\sigma_{-j})]
\end{equation*}
be seller $j$'s interim allocation probability.
BIC implies that $a_j$ is nonincreasing in cost and hence nonincreasing in the lower-tail rank.
The procurement payment identity and convex ironing give
\begin{equation*}
  \E[r_j]
  \ge\int_0^1\Psi_j(\sigma)a_j(\sigma)\dd \sigma
  =\E[\Psi_j(\sigma_j)q_j(\vct\sigma)].
\end{equation*}
For irregular costs, equality is attained when the allocation is constant on every seller ironing interval and the utility at the upper endpoint of the cost support is normalized to zero, interpreted as a limiting utility when that endpoint is not attained or is infinite.
Consequently, every feasible BIC and interim-IR buyer-side procurement auction satisfies
\begin{align*}
  \E[\pi_B(M,\vct r)]
  &\le
  \E\left[
    \sum_{(i,j)\in M}\bigl(v_i-\Psi_j(\sigma_j)\bigr)
  \right]\\
  &\le
  \E\left[
    \max_{M\in\cF}
    \sum_{(i,j)\in M}\bigl(v_i-\Psi_j(\sigma_j)\bigr)
  \right].
\end{align*}
The pointwise maximizer is monotone and constant on seller ironing intervals.
Normalized procurement payments make both inequalities tight, proving the one-sided optimality of GBOM.

\end{document}